\documentclass[aps,10.4pt,twocolumn, twoside]{revtex4}

\usepackage{amsmath, amsthm, amssymb,color,url,booktabs}  
\usepackage[pdftex]{hyperref} 

\newcommand{\nc}{\newcommand}
\nc{\rnc}{\renewcommand}

\newcommand{\bra}[1]{\left\langle #1\right|}
\newcommand{\ket}[1]{\left|#1\right\rangle}
\newcommand{\proj}[1]{\left|#1\right\rangle\left\langle #1\right|}

\DeclareMathOperator{\id}{id}

\DeclareMathOperator{\Par}{Par}
\DeclareMathOperator{\poly}{poly}

\DeclareMathOperator{\tr}{tr}

\DeclareMathOperator{\Sym}{Sym}

\DeclareMathOperator{\Sch}{Sch}

\def\be#1\ee{\begin{equation}#1\end{equation}}
\def\bea#1\eea{\begin{eqnarray}#1\end{eqnarray}}
\def\beas#1\eeas{\begin{eqnarray*}#1\end{eqnarray*}}
\def\ba#1\ea{\begin{align}#1\end{align}}
\def\bas#1\eas{\begin{align*}#1\end{align*}}
\def\bpm#1\epm{\begin{pmatrix}#1\end{pmatrix}}

\def\eq#1{(\ref{eq:#1})}
\def\eqs#1#2{(\ref{eq:#1}) and (\ref{eq:#2})}

\def\ra{\rightarrow}
\def\ot{\otimes}
\nc{\grad}{{\vec{\nabla}}}
\newcommand{\iso}[1]{\stackrel{#1}{\cong}}

\def\cD{\mathcal{D}}
\def\cE{\mathcal{E}}

\def\cP{\mathcal{P}}
\def\cQ{\mathcal{Q}}

\def\cS{\mathcal{S}}

\def\cU{\mathcal{U}}

\def\bbC{\mathbb{C}}
\DeclareMathOperator*{\E}{\mathbb{E}}

\def\benum{\begin{enumerate}}
\def\eenum{\end{enumerate}}
\def\bit{\begin{itemize}}
\def\eit{\end{itemize}}
\newcommand{\fig}[1]{Fig.~\ref{fig:#1}}

\newcommand{\lemref}[1]{Lemma~\ref{lem:#1}}

\nc{\todo}[1]{\textcolor{red}{todo: #1}}

\def\begsub#1#2\endsub{\begin{subequations}\label{eq:#1}\begin{align}#2\end{align}\end{subequations}}
\nc\qand{\qquad\text{and}\qquad}
\nc\mnb[1]{\medskip\noindent{\bf #1}}

\usepackage{amsthm}
\usepackage{amsmath,latexsym,amssymb,verbatim,enumerate,graphicx}
\usepackage{color}

\newtheorem{definition}{Definition} 
\newtheorem{prop}[definition]{Proposition}
\newtheorem{lemma}[definition]{Lemma}

\newtheorem{thm}[definition]{Theorem}

\newtheorem*{rep@theorem}{\rep@title}
\newcommand{\newreptheorem}[2]{%
\newenvironment{rep#1}[1]{%
 \def\rep@title{#2 \ref{##1} (restatement)}%
 \begin{rep@theorem}}%
 {\end{rep@theorem}}}
\makeatother

\newreptheorem{thm}{Theorem}
\newreptheorem{lem}{Lemma}

\def\ba#1\ea{\begin{align}#1\end{align}}
\def\ban#1\ean{\begin{align*}#1\end{align*}}

\def\benum{\begin{enumerate}}
\def\eenum{\end{enumerate}}

\def\squareforqed{\hbox{\rlap{$\sqcap$}$\sqcup$}}
\def\qed{\ifmmode\squareforqed\else{\unskip\nobreak\hfil
\penalty50\hskip1em\null\nobreak\hfil\squareforqed
\parfillskip=0pt\finalhyphendemerits=0\endgraf}\fi}
\def\endenv{\ifmmode\;\else{\unskip\nobreak\hfil
\penalty50\hskip1em\null\nobreak\hfil\;
\parfillskip=0pt\finalhyphendemerits=0\endgraf}\fi}

\def\id{{\operatorname{id}}}

\def\eps{\epsilon}

\mathchardef\ordinarycolon\mathcode`\:
\mathcode`\:=\string"8000
\def\vcentcolon{\mathrel{\mathop\ordinarycolon}}
\begingroup \catcode`\:=\active
  \lowercase{\endgroup
  \let :\vcentcolon
  }

\nc{\avg}[1]{\langle#1\rangle}

\nc{\smfrac}[2]{\mbox{$\frac{#1}{#2}$}} \nc{\Tr}{\operatorname{Tr}}
\nc{\ox}{\otimes} \nc{\dg}{\dagger} \nc{\dn}{\downarrow}
\nc{\lmax}{\lambda_{\text{max}}}
\nc{\lmin}{\lambda_{\text{min}}}

\nc{\csupp}{{\operatorname{csupp}}}
\nc{\qsupp}{{\operatorname{qsupp}}} \nc{\var}{\operatorname{var}}
\nc{\rar}{\rightarrow} \nc{\lrar}{\longrightarrow}

\nc{\Lip}{\operatorname{Lip}}
\nc{\Om}{\Omega}
\nc{\wt}[1]{\widetilde{#1}}

\nc{\glneq}{{\raisebox{0.6ex}{$>$}  \hspace*{-1.8ex} \raisebox{-0.6ex}{$<$}}}
\nc{\gleq}{{\raisebox{0.6ex}{$\geq$}\hspace*{-1.8ex} \raisebox{-0.6ex}{$\leq$}}}

\nc{\vholder}[1]{\rule{0pt}{#1}}
\nc{\wh}[1]{\widehat{#1}}
\nc{\h}[1]{\widehat{#1}}

\nc{\ob}[1]{#1}

\def\beq{\begin {equation}}
\def\eeq{\end {equation}}

\begin{document}

\title{{\Large Quantum Conditional Mutual Information, Reconstructed States, and State Redistribution}}

\author{Fernando G.S.L. Brand\~ao$^{1}$$^,$$^{5}$, Aram W. Harrow$^{2}$, Jonathan Oppenheim$^{3}$, Sergii Strelchuk$^{4}$}
\affiliation{$^{1}$ Department of Computer Science, University College London, Gower St, London WC1E 6BT, UK \\
$^{2}$ Center for Theoretical Physics, Massachusetts Institute of Technology,
Cambridge MA 02139, USA \\
$^{3}$ Department of Physics, University College London, Gower St, London WC1E 6BT, UK \\
$^{4}$ Department of Applied Mathematics and Theoretical Physics, University of Cambridge, Wilberforce Road, Cambridge CB3 0WA, UK
$^{5}$Quantum Architectures and Computation Group, Microsoft Research, Redmond, WA 98052-6399, USA
}

\begin{abstract}

We give two strengthenings of an inequality for the quantum conditional
mutual information of a tripartite quantum state recently proved by
Fawzi and Renner, connecting it with the ability to reconstruct the
state from its bipartite reductions. Namely, we show that the conditional
mutual information is an upper bound on the regularized relative entropy
distance between the quantum state and its reconstructed version. It is
also an upper bound for the measured relative entropy distance of the state
to its reconstructed version. The main ingredient of the proof is the
fact that the conditional mutual information is the optimal quantum
communication rate in the task of state redistribution.
\end{abstract}
\maketitle

\parskip .75ex

Quantum information theory is the successful framework describing the transmission and storage of information. It not only generalized all of the classical information-theoretic results but also developed a wealth of tools to analyze a number of scenarios beyond the reach of the latter, such as entanglement processing. One of the central quantities of the classical information theory that directly generalizes to quantum information is the conditional mutual information (CMI). For a tripartite state $\rho_{BCR}$ it is defined as

\begin{eqnarray}
&& I(C : R | B)_{\rho} \\  &:=& S(BC)_{\rho} + S(BR)_{\rho} - S(BCR)_{\rho} - S(B)_{\rho}, \nonumber
\end{eqnarray}
with $S(X)_{\rho} := - \tr(\rho_X \log \rho_X)$ as von Neumann entropy.  It measures the correlations of subsystems $C$ and $R$ relative to subsystem $B$. The fact the classical CMI is non-negative is a simple consequence of the properties of the probability distributions; the same fact for the quantum CMI is equivalent to a deep result of quantum information theory -- strong subadditivity of the von Neumann entropy \cite{LR73}. Naturally, this led to a variety of applications in different areas, ranging from quantum information theory \cite{DY08, CW04, BCY11} to condensed matter physics \cite{PH11, Kim13, Kit13}.

In the classical case, for every tripartite probability distribution $p_{XYZ}$,
\begin{equation} \label{classicalrelativeentropy}
I(X : Z | Y) = \min_{q \in \text{MC}} S(p || q),
\end{equation}
where $S(p || q) := \sum_i p_i \log(p_i/q_i)$ is the relative entropy and the minimum is taken over the set MC of all distributions $q$ such that $X - Y - Z$ form a Markov chain. Equivalently, the minimization in the right-hand side of Eq. (\ref{classicalrelativeentropy}) could be taken over $\Lambda \otimes \id_Z (p_{YZ})$, for reconstruction channels  $\Lambda : Y \rightarrow YX$. In particular, $I(X : Z | Y) = 0$ if, and only if, $X - Y - Z$ form a Markov chain (which is equivalent to the existence of a channel $\Lambda : Y \rightarrow YX$ such that $p_{XYZ} = \Lambda \otimes \id_Z (p_{YZ})$).

The class of tripartite quantum states $\rho_{BCR}$ satisfying $I(C : R | B)_{\rho} = 0$ has also been similarly characterized~\cite{HJPW04}: The $B$ subsystem can be decomposed as $B = \bigoplus_k B_{L, k} \otimes B_{R, k}$ (with orthogonal vector spaces $B_{L, k}\otimes B_{R, k}$) and the state written as
\begin{equation} \label{QMS}
\rho_{BCR} = \bigoplus_{k} p_k \rho_{CB_{L, k}} \otimes \rho_{B_{R, k} R}
\end{equation}
for a probability distribution $\{  p_k\}$ and states $\rho_{CB_{L, k}} \in C \otimes B_{L, k}$ and $ \rho_{B_{R, k} R} \in B_{R, k} \otimes R$. States of this form are called quantum Markov because in analogy to Markov chains, conditioned on the outcome of the measurement onto $\{ B_{L, k} \otimes B_{R, k} \}$, the resulting state on $C$ and $R$ is product. 

Paralleling the classical case, $\rho_{BCR}$ is a quantum Markov state
if, and only if, there exists a reconstruction channel $\Lambda : B
\rightarrow BC$ such that $\Lambda \otimes \id_R(\rho_{BR}) =
\rho_{BCR}$ \footnote{Ref.~\cite{HJPW04} shows that Markov states have
  perfect reconstruction maps; see Eq.~\eq{transpose} for an explicit formula.  The converse is straightforward.
  If $\Lambda \otimes \id_R(\rho_{BR}) =\rho_{BCR}$ then
$\Lambda(\rho_B)=\rho_{BC}$.  Using the recovery condition,
and the monotonicity of relative
entropy under applying first $(\Lambda\ot \id_R)$ and then $\tr_C$, we obtain
$S(\rho_{BCR} \| \rho_R \ot \rho_{BC}) = 
S((\Lambda \ot \id_R)\rho_{BR} \| (\Lambda \ot \id_R)(\rho_B \ot \rho_R))
 \leq S(\rho_{BR} \| \rho_B \ot \rho_R)
  \leq S(\rho_{BCR} \| \rho_{BC} \ot \rho_{R}).$
It then follows that
$I(C:R|B)= S(\rho_{BCR} \| \rho_{BC} \ot \rho_{R})- S(\rho_{BR} \|
\rho_B \ot \rho_R) = 0$.}.
Having generalized the definition of CMI, can we also retain the above
equivalence, with the set of quantum Markov 
states taking the role of Markov chains?  Surprisingly, it turns out that
this is not the case~\cite{ILW08} and it seems not to be possible to
connect states that are close to Markov states with states of small
conditional mutual information in a meaningful way (see, however,
\cite{BCY11, LW14b}).

Nonetheless, it might be possible to relate states with small conditional mutual information with those that can be approximately reconstructed from their bipartite reductions, i.e. such that $\Lambda \otimes \id_R(\rho_{BR}) \approx \rho_{BCR}$. Indeed, several conjectures appeared recently in this respect~\cite{Kim13, Zha13, OS13, BSW14}. 

A recent breakthrough result from Fawzi and Renner gives the first such connection. They proved the following inequality~\cite{FR14}: 

\begin{equation}  \label{FawziRennerInequality}
I(C : R | B)_{\rho} \geq \min_{\Lambda : B \rightarrow BC} S_{1/2}(\rho_{BCR} || \Lambda \otimes \id_R(\rho_{BR}))
\end{equation}
with $S_{1/2}(\rho || \sigma) := -2 \log F(\rho, \sigma)$ the
order-$1/2$ R\'{e}nyi relative entropy, where $F(\rho, \sigma) = \tr[(\sigma^{\frac{1}{2}} \rho \sigma^{\frac{1}{2}})^{\frac{1}{2}}]$ is the fidelity~\cite{Josza94}. It implies that if the conditional mutual information of $\rho_{BCR}$ is {small}, there exists a reconstructing channel $\Lambda : B \rightarrow BC$ such that $\Lambda \otimes \id_R(\rho_{BR})$ has {high} fidelity with $\rho_{BCR}$.

In this Letter, we prove a strengthened version of the Fawzi-Renner inequality. We also give a simpler proof of the inequality based on the task of state redistribution,~\cite{DY08} which gives an operational interpretation to the conditional mutual information.   

\vspace{0.5 cm}

\noindent \textbf{Result.} Let $S(\rho || \sigma) := \tr[\rho(\log \rho - \log \sigma)]$ be the quantum relative entropy of $\rho$ and $\sigma$. Define the measured relative entropy as 
\be
\mathbb{M}S(\rho || \sigma) = \max_{M \in {\cal M}} S(M(\rho) || M(\sigma)),
\ee
where ${\cal M}$ is the set of all quantum-classical channels $M(\rho) = \sum_k \tr(M_k \rho) \ket{k}\bra{k}$, with $\{ M_k \}$ a POVM and $\{ \ket{k} \}$ an orthonormal basis.

The main result of this Letter is the following:

\begin{thm} \label{main}
For every state $\rho_{BCR}$,
\begsub{mainthm}
& I(C:R|B)_{\rho} \\ 
\geq& \lim_{n \rightarrow \infty} \min_{\Lambda_n : B^n \rightarrow
  B^nC^n} 
\frac{1}{n} S(\rho_{BCR}^{\otimes n}|| \Lambda_n\otimes
\id_{R^n}(\rho_{BR}^{\otimes n})) 
\label{eq:main-asym} \\ 
 \geq&  \min_{\Lambda : B \rightarrow BC} 
\mathbb{M}S( \rho_{BCR} || \Lambda \otimes \id_R(\rho_{BR}) )
\label{eq:main-measured}\\
\geq &   \min_{\Lambda : B \rightarrow BC} 
S_{1/2}(\rho_{BCR} \| \Lambda \ot \id_R(\rho_{BR})).
\label{eq:main-FR}
\endsub
\end{thm}

Eq.~\eq{main-FR} is the Fawzi-Renner inequality
(Eq. (\ref{FawziRennerInequality})) and follows from
Eq.~\eq{main-measured} using the bound $S(\pi || \sigma)
\geq S_{1/2}(\pi || \sigma)$~\cite{FL13} and the fact that $\min_{M
  \in {\cal M}} F(M(\pi), M(\sigma)) = F(\pi, \sigma)$~\cite{FC94}.  Eq.~\eq{main-measured} also generalizes one side of
Eq. (\ref{classicalrelativeentropy}) to quantum states, implying that
it is optimal at least for classical states $\rho$.

Our lower bound provides a substantial improvement over the original
Fawzi-Renner bound even for classical states. To see this, consider
the classically correlated state $\rho_{CBR} = \rho_{CR}\otimes
{\frac{\mathbb{I}_B}{d_B}}$ with $d:=d_C=d_R$ and $\rho_{CR} =
(1-\eps)|00\rangle\langle 00|_{CR} + \frac{\eps}{d-1}
\sum_{k=1}^{d-1} |kk\rangle\langle
kk|_{CR}$. 
Then Eq.~\eqref{eq:main-measured} becomes $\mathbb{M}S( \rho_{BCR} || \sigma_{BC}\otimes \rho_{R})$, where $\sigma_{BC}$ depends on the channel $\Lambda$ that minimizes Eq.~\eqref{eq:main-measured}. The measured relative entropy is equal to the ordinary classical relative entropy between the distribution $p_Bp_{CR}$ (generated from $\rho_{BCR}$) and the product distribution $q_{BC}p_R$ (generated from $\sigma_{BC}\otimes\rho_R$) optimized over all quantum-classical channels. Observing that $p_{CR}$ is maximally correlated whereas $q_Cp_B$ is the product distribution irrespective of $\Lambda$, Eq.~\eqref{eq:main-measured} equals to $I(C:R)\approx\eps\log{\left(d-1\right)}$. The corresponding Fawzi-Renner
bound~\eqref{eq:main-FR} becomes $-\log
F(\rho_{CR},\rho_C\otimes\rho_R)\le -\log(1-\epsilon)\approx
\epsilon$. Thus, the lower bound~\eqref{eq:main-measured} is optimal
for classical states. 

Another application of our result is the well-known problem of classification of the short-range entangled states studied by Kitaev~\cite{kitaev_topological}. Defining such a class of states is nontrivial and one of the natural properties to be required is the ability generate them locally: There must exist a $O(1)$ quantum circuit that generates the designated state from a product state. In particular, one sees that states with low conditional mutual information can be generated from the product states according to the Fawzi-Renner bound. Our result improves the lower bound when we quantify the distance between the states using measured relative entropy. 

Li and Winter conjectured in \cite{LW14b} that Eq.~(\ref{eq:main-measured}) could be strengthened to have the relative entropy  in the right-hand side (instead of the measured relative entropy). We leave this as an open question, but we note that Eq. (\ref{eq:main-asym}) shows that an asymptotic version of the conjectured inequality does hold true.  

\hspace{0.5 cm}

\noindent \textbf{Proof of Theorem \ref{main}}: The main tool in the proof will be the state redistribution protocol of Devetak and Yard~\cite{DY08, YD09, Opp08} which gives an operational meaning for the conditional mutual information as twice the optimal quantum communication cost of the protocol. Consider the state $\ket{\psi}_{ABCR}^{\otimes n}$ shared by two parties (Alice and Bob) and the environment (or reference system). Alice has $A^nC^n$ (where we denote $n$ copies of $A$ by $A^n$ and likewise for $C$,$B$ and $R$), Bob has $B^n$, and $R^n$ is the reference system. In state redistribution, Alice wants to redistribute the $C^n$ subsystem to Bob using preshared entanglement and quantum communication. 

It was shown in~\cite{DY08, YD09} that using preshared entanglement Alice can send the $C^n$ part of her state to Bob, transmitting approximately $(n/2)I(C:R|B)$ qubits in the limit of a large number of copies $n$. More precisely:

\begin{lemma} [State Redistribution Protocol \cite{DY08, YD09}] \label{lemredistribution}
For every $\ket{\psi}_{ABCR}$ there exist completely positive trace-preserving encoding maps $\cE_n : A^nC^n X_n \rightarrow A^n  G_n$ and decoding maps $\cD_n : B^nG_n Y_n  \rightarrow B^n C^n$ such that 
\begin{equation} \label{correctenessstateredistribution}
\lim_{n \rightarrow \infty} \Vert  \cD_n \circ \cE_n ( \ket{\psi}\bra{\psi}_{ABCR}^{\otimes n} \otimes \Phi_{X_nY_n}) - \ket{\psi}\bra{\psi}_{ABCR}^{\otimes n}  \Vert_1 = 0
\end{equation}
and
\begin{equation}
\lim_{n \rightarrow \infty} \frac{\log \dim(G_n)}{n} = \frac{1}{2}I(C : R | B)_{\rho},
\end{equation}
where $\rho_{BCR} := \tr_{A}(\ket{\psi}\bra{\psi}_{ABCR})$ and $\Phi_{X_nY_n}$ is a maximally entangled state shared by Alice (who has $X_n$) and Bob (who has $Y_n$), and $\Vert.\Vert_1$ denotes the trace norm.
\end{lemma}

We split the proof of Theorem~\ref{main} into the proof of Proposition \ref{prop:asym} and Eq.~\eqref{eq:measured} below. 

Proposition~\ref{prop:asym} follows from the state redistribution protocol outlined above. The main idea is the following: Suppose that in the state redistribution protocol Bob does not receive any quantum communication from Alice, but instead he "mocks" the communication (locally preparing the maximally mixed state in $G_n$) and applies the decoding map $\cD_n$. It will follow that even though the output state might be very far from the target one, the relative entropy per copy of the output state and the original one cannot be larger than twice the amount of communication of the protocol (which is given by the conditional mutual information). 

\begin{prop}   \label{prop:asym}
For every state $\rho_{BCR}$, 
\begin{eqnarray}
&&I(C : R | B)_{\rho} 
\label{eq:prop-asym}
\\ &\geq& \lim_{n \rightarrow \infty} \min_{\Lambda : B^n \rightarrow B^nC^n} \frac{1}{n} S(\rho_{BCR}^{\otimes n}|| \Lambda \otimes \id_{R^n}(\rho_{BR}^{\otimes n})).  \nonumber 
\end{eqnarray}
\end{prop}

\begin{proof}
Let $\ket{\psi}_{ABCR}$ be a purification of $\rho_{BCR}$. Consider the state redistribution protocol for sending $C$ from Alice (who has $AC$) to Bob (who has $B$). Let $\phi_{G_n Y_n A^n B^n R^n} := \cE_n \otimes \id_{B^nR^nY_n} (\ket{\psi}\bra{\psi}_{ABCR}^{\otimes n} \otimes \Phi_{X_n Y_n})$ be the state after the encoding operation. 

Using the operator inequality $\pi_{MN} \leq \dim(M) I_M \otimes \pi_N$, valid for every state $\pi_{MN}$, we find
\begin{equation} \label{opineqphi}
\phi_{G_n Y_n B^n R^n} \leq \dim(G_n)^2 \tau_{G_n} \otimes \tau_{Y_n} \otimes \rho_{B R}^{\otimes n} . 
\end{equation}
with $\tau_{Y_n}, \tau_{G_n}$ as the maximally mixed state on $Y_n$ and $G_n$, respectively. We used $\phi_{Y_n B^n R^n} = \tau_{Y_n} \otimes \rho_{B R}^{\otimes n}$, which holds true since $\cE_n$ only acts nontrivially on $A^n C^n X_n$

Let $\cD_n : G_n Y_n B^n \rightarrow B^n C^n$ be the decoding operation of Bob in state redistribution (see Lemma \ref{lemredistribution}) and define $\tilde{D}_n := (1 - 2^{-n})\cD_n + 2^{-n} \Lambda_{\text{dep}}$, with $\Lambda_{\text{dep}}$ the depolarizing channel mapping all states to the maximally mixed. Since $\tilde \cD_n$ is completely positive, using Eq. (\ref{opineqphi}) we get
\begin{eqnarray}
&&(\tilde \cD_n \otimes \id_{R^n}) (\tau_{G_n} \otimes \tau_{Y_n} \otimes \rho_{BR}^{\otimes n}) \\ &\geq& \dim(G_n)^{-2} (\tilde \cD_n \otimes \id_{R^n}) (\phi_{G_n Y_n B^n R^n}).  \nonumber
\end{eqnarray}

From the operator monotonicity of the  $\log$ (see Lemma~\ref{opmon} in the Supplemental Material),
\begin{eqnarray}
&& S(\rho_{BCR}^{\otimes n} || (\tilde \cD_n \otimes \id_{R^n}) (\tau_{G_n} \otimes \tau_{Y_n} \otimes \rho_{BR}^{\otimes n})) \\ &\leq& S(\rho_{BCR}^{\otimes n} || (\tilde \cD_n \otimes \id_{R^n}) (\phi_{G_n Y_n B^n R^n})) + 2\log(\dim(G_n)).  \nonumber
\end{eqnarray}

Eq. (\ref{correctenessstateredistribution}) gives 
\begin{equation}
\lim_{n \rightarrow \infty} \Vert \rho_{BCR}^{\otimes n} - (\tilde \cD_n \otimes \id_{R^n}) (\phi_{G_n Y_n B^n R^n}) \Vert_1 = 0.
\end{equation}

Because $(\tilde \cD_n \otimes \id_{R^n}) (\phi_{G_n Y_n B^n R^n}) = (1 - 2^{-n})(\cD_n \otimes \id_{R^n})(\phi_{G_n Y_n B^n R^n}) + 2^{-n} \tau_{BC}^{\otimes n} \otimes \rho_R^{\otimes n}$ (with $\tau_{BC}$ the maximally mixed state on $BC$), Lemma~\ref{AudenaertE} in the Supplemental Material gives 
\begin{equation}
\lim_{n \rightarrow \infty} \frac{1}{n} S(\rho_{BCR}^{\otimes n} || (\tilde \cD_n \otimes \id_{R^n}) (\phi_{G_n Y_n B^n R^n})) = 0,
\end{equation}
and so
\begin{eqnarray}
&&I(C : R | B)_{\rho} \\ &=& 2\lim_{n \rightarrow \infty} \frac{\log(\dim(G_n))}{n} \nonumber \\
&\geq&  \lim_{n \rightarrow \infty} \min_{\Lambda_n : B^n \rightarrow B^n C^n} \frac{1}{n} S(\rho_{BCR}^{\otimes n} || (\Lambda_n \otimes \id_{R^n}) (\rho_{BR}^{\otimes n})). \nonumber 
\end{eqnarray}
\end{proof}

Even though we do not know whether 

\begin{eqnarray}
&& \lim_{n \rightarrow \infty} \min_{\Lambda : B^n \rightarrow B^nC^n} \frac{1}{n} S(\rho_{BCR}^{\otimes n}|| \Lambda \otimes \id_{R^n}(\rho_{BR}^{\otimes n})) \nonumber \\
&\stackrel{?}{\geq} & \min_{\Lambda : B \rightarrow BC} S(\rho_{BCR} || \Lambda \otimes \id_{R}(\rho_{BR})),
\end{eqnarray}

it turns out that a similar inequality holds true if we replace the relative entropy by its measured variant (see Sec. B in the Supplemental Material): For every state $\rho_{BCR}$ one has

\begin{eqnarray} \label{eq:measured}
&& \lim_{n \rightarrow \infty} \min_{\Lambda : B^n \rightarrow B^nC^n} \frac{1}{n} S(\rho_{BCR}^{\otimes n}|| \Lambda \otimes \id_{R^n}(\rho_{BR}^{\otimes n})) \nonumber \\ &\geq&  \min_{\Lambda : B \rightarrow BC}  \mathbb{M}S( \rho_{BCR} || \Lambda \otimes \id_R(\rho_{BR}) ).\label{eq:prop2statmt}
\end{eqnarray}

\noindent \textbf{Discussion and Open Problems.} The main result of this Letter, on one hand, and Theorem 4 of Ref. \cite{ILW08}, on the other hand, give
\begin{eqnarray} \label{lowerandupperCMI}
&& \min_{\sigma \in \text{QMS}} S(\rho_{BCR} || \sigma_{BCR}) \\ &\geq& I(C : R | B) \geq  \min_{\Lambda : B \rightarrow BR} \mathbb{M} S(\rho_{BCR} || \Lambda \otimes \id_R (\rho_{BR})), \nonumber  
\end{eqnarray}
with QMS the set of quantum Markov states given by Eq. (\ref{QMS}). For probability distributions the lower and upper bounds in Eq. (\ref{lowerandupperCMI}) coincide, giving Eq. (\ref{classicalrelativeentropy}). However, in the quantum case, the two can be very far from each other. 

An interesting question is whether we can also have equality in the
quantum case when minimizing over the set of reconstructed states. In
particular we can ask whether Eq.~\eq{prop-asym} holds with
equality. It turns out that this is false and can be disproved using
pure states of dimension $2\times 2\times 2$ and the transpose channel, defined for
a tripartite state $\rho_{BCR}$ as 
\begin{equation}
T(\pi) := \sqrt{\rho_{BC}} \left( {\rho_{B}^{-1/2}} \pi
  {\rho_{B}^{-1/2}} \otimes \id_C \right) \sqrt{\rho_{BC}}.
\label{eq:transpose}
\end{equation}

\begin{figure}
\includegraphics[width=0.5\textwidth]{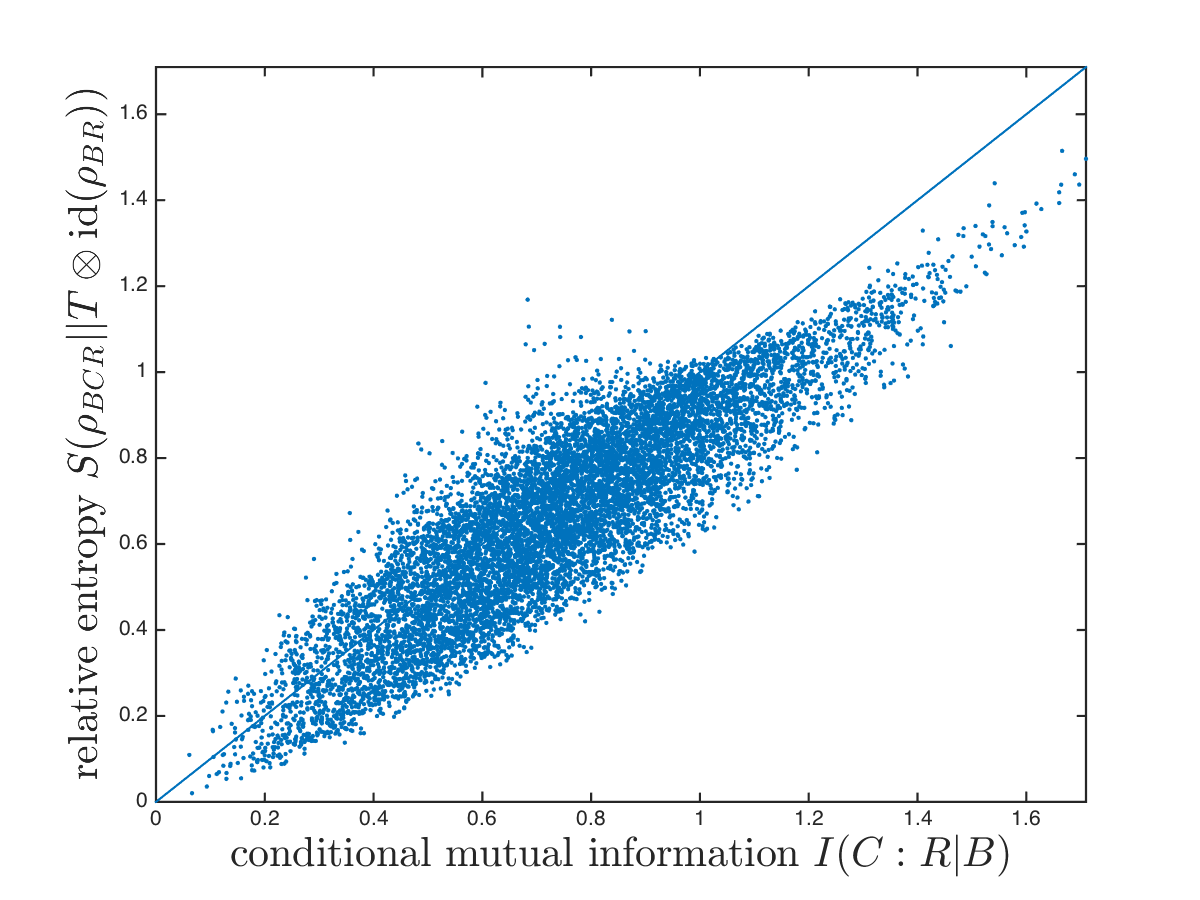}
\caption{Counterexamples for the case of equality in Eq.~\eq{prop-asym}: Conditional mutual information against the reconstructed relative entropy using the transpose channel. The sample consists of 10000 random pure states of dimension $2\times 2\times 2$.\label{fig:numerics}}
\end{figure}

In \fig{numerics} we plot the conditional mutual information
against the reconstructed relative entropy using the transpose channel
(i.e. $S(\rho_{BCR} \| T_B \ot \id_R(\rho_{BR}))$) for 10,000 randomly
chosen pure states of dimension $2\times 2\times 2$.  We see that for
roughly 73\% of the points, the relative entropy is strictly smaller
than the conditional mutual information when using the transpose
channel.  Since any particular reconstruction map also puts an upper
bound on the minimum relative entropy, Eq. \eq{prop-asym} must sometimes
be a strict inequality.  Similar numerical results were found in an
unpublished early version of~\cite{LW14b}.

In the proof of Theorem \ref{main} we were not able to give an explicit optimal reconstruction map. In the context of approximate recovery for pure states, the transpose channel is optimal up to a square factor~\cite{BK02} (using the fidelity as a figure of merit). We could ask whether the same holds for mixed states.  

Another interesting open problem is whether we can improve the lower bound in Eq. (\ref{lowerandupperCMI}) to have the relative entropy, instead of the measured relative entropy. Proposition \ref{prop:asym} and \lemref{ps-channel} in the Supplemental Material show that the result would follow from the following conjectured inequality: Given a state $\rho$, a convex closed set of states ${\cal S}$, and a measure $\mu$ with support only on ${\cal S}$,
\begin{equation}
\liminf_{n \rightarrow \infty} \frac{1}{n} S \left(\rho^{\otimes n} || \int \mu(d \sigma) \sigma^{\otimes n} \right) \stackrel{?}{\geq} \min_{\sigma \in {\cal S}} S(\rho || \sigma). 
\end{equation}

The case when $\rho_{BR} = \rho_B\ot\rho_R$ was recently proven in~\cite{HT14}. 
We can also easily prove the inequality classically, using hypothesis
testing, which is universal for the alternative hypothesis. However, since 
there is no quantum hypothesis test universal for the alternative hypothesis~\cite{qSanov} for general sets ${\cal S}$, we leave the inequality in the quantum case as an open problem for future work. 

\mbox{}

\noindent \textbf{Acknowledgements:} F.G.S.L.B. and J.O. thank EPSRC for
financial support. A.W.H. was funded by NSF Grants No. CCF-1111382 and CCF-1452616, ARO Contract No.
W911NF-12-1-0486 and Visiting Professorship VP2-2013-041 from the Leverhulme Trust.
 S.S. acknowledges the support of Sidney Sussex College. We thank Omar
 Fawzi, Ke Li, Brian Swingle, Marco Tomamichel, Mark
 Wilde, and Lin Zhang for helpful comments and the Newton Institute for their
 hospitality while some of this research was conducted. 

\section*{SUPPLEMENTAL MATERIAL}
\subsection{Auxiliary lemmas}
\begin{lemma} \label{opmon}
If $\pi\leq 2^{\lambda} \sigma$, then $S(\rho || \pi) \geq S(\rho || \sigma) - \lambda$. 
\end{lemma}
The proof of the lemma follows directly from the operator monotonicity of the log function.\qed

Lemma \ref{AudenaertE} is due to Audenaert and Eisert:

\begin{lemma} [Theorem 3 of \cite{AE05}] \label{AudenaertE}
For all states $\rho$  and $\sigma$ on a $d$-dimensional Hilbert space, with $T = \Vert \rho - \sigma \Vert_1$ and $\beta = \lambda_{\min}(\sigma)$,
\be
S(\rho || \sigma) \leq T \log(d) + \min \left( - T \log T, \frac{1}{e}   \right) - \frac{T \log \beta}{2}.
\ee
\end{lemma}

The next lemma is due to M. Piani. It will suffice to state it for the
case where $\mathbb{M}$ is the set of all measurements. 


\begin{lemma} \label{Piani}
[Theorem 1 of \cite{Piani09}] Consider two systems $X$ and $Y$ with joint Hilbert space ${\cal H}_X \otimes {\cal H}_Y$, and a convex reference set $K$. Suppose the reference set $K$ is such that for all POVM elements $M_i$ and all $\sigma_{XY} \in K$, $\tr_{X}(M_i^X  \sigma_{XY}) \in P$ (up to normalization). Then for every $\rho_{XY}$,
\begin{eqnarray}
&& \min_{\sigma_{XY} \in K} S(\rho_{XY} || \sigma_{XY})
  \\ &\geq& \min_{\sigma_X \in K} \mathbb{M}S(\rho_X || \sigma_X) 
+ \min_{\sigma_Y \in K} S(\rho_Y || \sigma_Y). \nonumber
\end{eqnarray}
\end{lemma}

The following Lemma is due to Fawzi and Renner~\cite{FR14} who stated it
in a slightly more general form. It was used in their proof of
Eq. (4) in the main text.  Below is a very similar, but
somewhat shorter, proof.

\begin{lemma}\label{lem:postselect}
Suppose $\rho_{X^nY^n}$ satisfies $\rho_{X^n} = \tau_{X^n}$ and
$[\rho, P_{XY}(\pi)]=0$ for all $\pi\in \cS_n$.  Then there exists a
measure $\mu$ over states $\sigma_{XY}$ (independent of $\rho$) with
each $\sigma_X = \tau_X$ 
and
\be
\rho_{X^nY^n} \leq n^{O(d_X^2d_Y^2)} \int \sigma^{\ot n} \mu(d\sigma),
\label{eq:constrained-post}\ee
where $d_X, d_Y$ are the dimensions of $X$ and $Y$.
\end{lemma}

\begin{proof}
First purify $\rho$ to a state $\ket{\rho}_{X^nY^nZ^n}\in\Sym^n(XYZ)$ (the symmetric subspace in $(XYZ)^{n}$)
with $d_Z = d_Xd_Y$ using Lemma 4.2.2 of \cite{Ren05}.  For $V$ an
isometry from $X \ra YZ$, define 
\be \ket{\sigma(V)}_{XYZ} = \frac{1}{\sqrt{d_X}} \sum_{i=1}^{d_X}
\ket{i} V\ket i\ee
and $\sigma(V) = \proj{\sigma(V)}$.  Observe that $\sigma(V)_X =
\tau_X$.  We will show that
\be
\proj{\rho} \leq n^{O(d_X^2d_Y^2)}
\int \sigma(V)^{\ot n} \mu(d\sigma),
\label{eq:purified-post}\ee
which will imply Eq. \eq{constrained-post}.

Our strategy will be to expand both sides of Eq. \eq{purified-post} in the
Schur basis.  Schur duality uses the following notation:
\be (\bbC^d)^{\ot n} \iso{\cU_{d} \times \cS_n}
\bigoplus_{\lambda\in\Par(n,d)} \cQ_\lambda^{d} \hat{\ot} \cP_\lambda.
\label{eq:schur-decomp}\ee
This is explained in detail in \cite{Har05}, but briefly,
$\Par(n,d)$ denotes the set of partitions of $n$ into $\leq d$
parts, $\cQ_{\lambda}^{d}$ is an irrep of the unitary group
$\cU_{d}$,
$\cP_\lambda$ is an irrep of the symmetric group $\cS_n$, $\hat\ot$
means that we interpret the tensor product as an irrep of $\cU_d
\times \cS_n$,  and
$\iso{\cU_{d} \times \cS_n}$ means that the isomorphism respects this
representation structure.
Let $U_{\Sch}$ denote the unitary transform realizing the
isomorphism in Eq. \eq{schur-decomp}.  We can write
\begin{multline} (U_{\Sch}^X \ot U_{\Sch}^{YZ})\ket\rho
= \\
 \sum_{\substack{\lambda_1\in\Par(n,d_X) \\
    \lambda_2\in\Par(n,d_Yd_Z)}}
 c_{\lambda_1,\lambda_2}
\ket{\lambda_1}_X  \ket{\lambda_2}_{YZ} 
\ket{\chi_{\lambda_1,\lambda_2}} \ket{\theta_{\lambda_1,\lambda_2}},
\label{eq:rho-decomp-first}\end{multline}
where $\sum_\lambda |c_\lambda|^2=1$,
$\ket{\chi_{\lambda_1,\lambda_2}},\ket{\theta_{\lambda_1,\lambda_2}}$
are arbitrary unit vectors in 
$\cQ_{\lambda_1}^{d_X}\ot\cQ_{\lambda_2}^{d_Yd_Z}$ and 
$\cP_{\lambda_1}\ot\cP_{\lambda_2}$ respectively.
However, the permutation invariance and Schur's Lemma mean that (following arguments
along the lines of Section 6.4.1 of \cite{Har05}) the only nonzero terms
have $\lambda_1=\lambda_2$ and $\ket{\theta_{\lambda,\lambda}} =: \ket{\Phi_\lambda}$ is the unique permutation-invariant state in
$\cP_\lambda \ot \cP_\lambda$.   Thus we can (using $d_X \leq d_Yd_Z$) rewrite 
Eq. \eq{rho-decomp-first} as 
\be (U_{\Sch}^X \ot U_{\Sch}^{YZ})\ket\rho
= \sum_{\lambda\in\Par(n,d_X)} c_{\lambda}
\ket{\lambda}_X  \ket{\lambda}_{YZ}  
\ket{\chi_{\lambda}} \ket{\Phi_{\lambda}}.
\label{eq:rho-decomp-second}\ee
To calculate $c_\lambda$ we use the
fact that $\rho_{X^n} = \tau_{X^n}$.  Thus measuring the irrep label
should yield outcome $\lambda$ with probability $\dim\cP_\lambda \dim
\cQ_{\lambda}^{d_X} / d_X^n$, and we have
\be c_\lambda = \sqrt{\frac{\dim\cP_\lambda
    \dim\cQ_{\lambda}^{d_X}}{d_X^n}}.\label{eq:c-lambda}\ee

A similar argument implies that
\begin{multline} (U_{\Sch}^X \ot U_{\Sch}^{YZ})\ket{\sigma(V)}^{\ot n}
= \\ \sum_{\lambda\in\Par(n,d_X)} c_{\lambda}
\ket{\lambda}_X \ket{\lambda}_{YZ} 
\ket{\chi_{\lambda}(V)} \ket{\Phi_{\lambda}},
\label{eq:sigma-decomp}\end{multline}
for some states $\ket{\chi_{\lambda}(V)}$.  The coefficients
$c_\lambda$ are the same as in Eq. \eq{c-lambda} because $\sigma(V)^{\ot
  n}_{A^n} = \tau_{A^n}$. 
Averaging $\sigma(V)^{\ot n}$ over all isometries $V$ yields a state
that commutes with $(U_X \ot I_{YZ})^{\ot n}$ and $(I_X \ot
U_{YZ})^{\ot n}$ for all $U_X \in \cU_{d_X}, U_{YZ} \in \cU_{d_Yd_Z}$.  Thus
\begin{multline}
(U_{\Sch}^X \ot U_{\Sch}^{YZ}) \E_V[\sigma(V)^{\ot n}] 
(U_{\Sch}^X \ot U_{\Sch}^{YZ})^\dag = \\
\sum_{\lambda\in\Par(n,d_X)} |c_{\lambda}|^2
\proj{\lambda,\lambda} 
\ot \tau_{\cQ_\lambda^{d_X}}\ot \tau_{\cQ_\lambda^{d_YdZ}}
\ot \proj{\Phi_{\lambda}}.
\label{eq:E-sigma-decomp}\end{multline}

It follows from \eqs{rho-decomp-second}{E-sigma-decomp} that
\ban
 \proj{\rho} &
\leq (\max_\lambda \dim \cQ_\lambda^{d_X}
\dim \cQ_\lambda^{d_Yd_Z}) \E_V[\sigma(V)^{\ot n}] 
\\ & = \binom{d_X + n-1}{n}\binom{d_Yd_Z + n-1}{n}
\E_V[\sigma(V)^{\ot n}] 
\\ &\leq n^{d_X}n^{d_Yd_Z}\E_V[\sigma(V)^{\ot n}] 
\\ & =n^{d_X^2d_Y^2+d_X}\E_V[\sigma(V)^{\ot n}].
\ean
\end{proof}

\subsection{The measured entropy lower bound}
To prove the inequality stated in the Eq. (17) of the main text, we first prove the following lemma which is a version of the postselection technique of~\cite{CKR09} for quantum operations. 

\begin{lemma}~\label{lem:ps-channel}
For every permutation-invariant quantum operation $\Lambda : B^n \rightarrow B^nC^n$ and every state $\pi
\in B^n$,
\begin{equation} \label{postselectionqop}
\Lambda(\pi) \leq \poly(n) \int  {\cal E}^{\otimes n}(\pi) \mu(d{\cal E}),
\end{equation}
where $\mu$ is a measure over quantum operations ${\cal E}: B \rightarrow BC$.
\end{lemma}

\begin{proof}
  Since $\Lambda : B^n \rightarrow B^nC^n$ is permutation-invariant,
  it follows that its Jamiolkowski state
  $J_{\Lambda} \in {\cal D}((\overline{B}^n \otimes B^n \otimes C^n)$
  (with $\overline{B} \cong B$) is permutation-invariant.  We now
  apply \lemref{postselect} in the Supplemental Material to find a distribution $\mu$ over
  $\sigma\in\cD(\overline{B} \ot B \ot C)$ with \be J_\Lambda \leq
  \poly(n) \int \sigma^{\ot n} \mu(d\sigma),\ee and each
  $\sigma_{\bar B} = \tau_{\bar B}$.  This latter condition means that
  each $\sigma$ can be also thought of as $J_{\cE}$ for some
  $\cE: B \ra BC$. We complete the proof using the relation:
\begin{eqnarray}
&&\tr_{\overline{B}^n} ( (\pi^T \otimes I_{B^nC^n}) J_{\Lambda}) \\ &\leq& \poly(n) \int \tr_{\overline{B}^n} ( (\pi^T \otimes I_{B^nC^n}) J_{{\cal E}}^{\otimes n}) \mu(d {\cal E}), \nonumber
\end{eqnarray}
and the fact that $\tr_{\overline{B}^n} ( (\pi^T \otimes I_{B^nC^n}) J_{\Lambda})  = \Lambda(\pi) / \text{dim}(B)^n$ and $\tr_{\overline{B}^n} ( (\pi^T \otimes I_{B^nC^n}) J_{{\cal E}}^{\otimes n}) = {\cal E}^{\otimes n}(\pi) / \text{dim}(B)^n$.
\end{proof}

We now turn to proving the measured entropy lower bound:

\begin{prop}[Eq. (17) in the main text]~\label{prop:measured}
For every state $\rho_{BCR}$ one has
\begin{eqnarray} 
&& \lim_{n \rightarrow \infty} \min_{\Lambda : B^n \rightarrow B^nC^n} \frac{1}{n} S(\rho_{BCR}^{\otimes n}|| \Lambda \otimes \id_{R^n}(\rho_{BR}^{\otimes n})) \nonumber \\ &\geq&  \min_{\Lambda : B \rightarrow BC}  \mathbb{M}S( \rho_{BCR} || \Lambda \otimes \id_R(\rho_{BR}) ).\label{prop2statmt}
\end{eqnarray}
\end{prop}

\begin{proof}

For $\Lambda : B^n \rightarrow B^n C^n$, define
\begin{equation}
\tilde{\Lambda}(\omega) := 
\frac{1}{n!} \sum_{\pi \in S_n} P_{BC}(\pi) \Lambda( P^{\cal y}_B(\pi)
\omega
 P_B(\pi)) P_{BC}(\pi)^{\cal y},
\end{equation}
with $P_{X}(\pi)$ a representation of a permutation $\pi$ from $S_n$
(symmetric group of order $n$) in $X^{\otimes n}$ such that $P_X(\pi)
\ket{a_1, \ldots, a_n} = \ket{a_{\pi^{-1}(1)}, \ldots,
  a_{\pi^{-1}(n)}}$. Let Sym be the set of all permutation-invariant
quantum operations, i.e. all $\Lambda$ such that $\Lambda = \tilde\Lambda$.

Using Proposition~\ref{prop:asym} in the main text and the fact that the relative entropy is doubly convex we obtain~\footnote{In more detail: $S(\rho_{BCR}^{\otimes n} || \Lambda (\rho_{BR}^{\otimes n} )) = \mathop{\mathbb{E}}_{\pi \in S_n}$  $S( P_{CBR}^{\pi} \rho_{BCR}^{\otimes n} (P_{CBR}^{\pi})^{\cal y} ||  P_{CBR}^{\pi} \Lambda ((P_{BR}^{\pi})^{\cal y} \rho_{BR}^{\otimes n} P_{BR}^{\pi}) (P_{CBR}^{\pi})^{\cal y})$ $= \mathop{\mathbb{E}}_{\pi \in S_n} S(\rho_{BCR}^{\otimes n} || P_{CB}^{\pi} \Lambda ((P_{B}^{\pi})^{\cal y} \rho_{BR}^{\otimes n} P_{B}^{\pi}) (P_{CB}^{\pi})^{\cal y}) \geq$ $ S\left ( \rho_{BCR}^{\otimes n}  || \mathop{\mathbb{E}}_{\pi \in S_n}  P_{CB}^{\pi} \Lambda ( (P_{B}^{\pi})^{\cal y} \rho_{BR}^{\otimes n} P_{B}^{\pi}) (P_{CB}^{\pi})^{\cal y} \right)$, with $P_{CBR}^{\pi} := P_{CBR}(\pi)$.}
\begin{eqnarray}
&& \lim_{n \rightarrow \infty} \min_{\Lambda : B^n \rightarrow B^nC^n} \frac{1}{n} S(\rho_{BCR}^{\otimes n}|| \Lambda \otimes \id_{R^n}(\rho_{BR}^{\otimes n})) \nonumber  \\ &\geq& \lim_{n \rightarrow \infty} \min_{  \substack{\Lambda : B^n \rightarrow B^nR^n \\   \Lambda \in \text{Sym}}   } \frac{1}{n}S(\rho_{BCR}^{\otimes n}|| \Lambda \otimes \id_{R^n}(\rho_{BR}^{\otimes n})).  \nonumber
\end{eqnarray}

Lemma~\ref{lem:ps-channel} gives that for every $\Lambda_n : B^n \rightarrow B^nC^n \in$ Sym,
\begin{eqnarray}
&& (\Lambda_n \otimes \id_{R^n})(\rho_{BR}^{\otimes n}) \\ &\leq& \poly(n) \int ({\cal E} \otimes \id_R (\rho_{BR}))^{\otimes n} \mu_n(d{\cal E}), \nonumber
\end{eqnarray}
with $\mu(d{\cal E})$ a measure over quantum operations ${\cal E} : B \rightarrow BC$. Using the previous equation and the operator monotonicity of the $\log$ (see Lemma \ref{opmon} above),
\begin{eqnarray}   \label{eq21}
&& \lim_{n \rightarrow \infty} \min_{\Lambda : B^n \rightarrow B^nC^n} \frac{1}{n} S(\rho_{BCR}^{\otimes n}|| \Lambda \otimes \id_{R^n}(\rho_{BR}^{\otimes n}))  \\ &\geq& \lim_{n \rightarrow \infty} \min_{\mu_n} \frac{1}{n}S\left(\rho_{BCR}^{\otimes n} \hspace{0.1 cm} \Vert \int ({\cal E} \otimes \id_R (\rho_{BR}))^{\otimes n} \mu_n(d{\cal E})\right). \nonumber
\end{eqnarray}

To complete the proof we make use of Lemma~\ref{Piani} above. Consider the state $\rho_{BCR}^{\otimes n}$ and let $X$ be the first copy of $\rho_{BCR}$ in the tensor product and $Y$ the remaining $\rho_{BCR}^{\otimes n-1}$. Define
\be
K = \bigcup_{k \in \mathbb{N}} \left(\text{conv} \{ ( {\cal E} \otimes \id_R)( \rho_{BR} )^{\otimes k} \hspace{0.1 cm} : \hspace{0.1 cm} {\cal E} : B \rightarrow BC \} \right),
\ee
i.e. the convex hull of tensor products of reconstructed states. It is easy to check that $K$ satisfies the assumption of Lemma~\ref{Piani} above. Therefore:
\begin{eqnarray}
&& \min_{\mu_n}  S\left(\rho_{BCR}^{\otimes n} \hspace{0.1 cm} \Vert \int ({\cal E} \otimes \id_R (\rho_{BR}))^{\otimes n} \mu_n(d{\cal E})\right)  \\
&\geq& \min_{\mu}   \mathbb{M} S\left(\rho_{BCR} \hspace{0.1 cm} \Vert \int ({\cal E} \otimes \id_R (\rho_{BR})) \mu(d{\cal E})\right) \nonumber \\  &+&   \min_{\mu_{n-1}}  S\left(\rho_{BCR}^{\otimes n-1} \hspace{0.1 cm} \Vert \int ({\cal E} \otimes \id_R (\rho_{BR}))^{\otimes n-1} \mu_{n-1}(d{\cal E})\right).  \nonumber
\end{eqnarray}
Iterating the equation above $n$ times gives Eq.~\eqref{eq:prop2statmt}.
\end{proof}

\end{document}